\documentclass[onecollarge]{svjour2} 

\smartqed  

\usepackage{graphicx}

\usepackage{amsmath}
\usepackage{amssymb}
\usepackage{soul}
\usepackage{enumerate}
\usepackage{color}

\newcommand{\Gr}{\mathcal{G}}

\begin{document}

\title{The Social Climbing Game
}


\author{Marco Bardoscia  \and Giancarlo De Luca \and \\ Giacomo Livan \and Matteo Marsili \and Claudio J. Tessone
}


\institute{M. Bardoscia \and G. Livan \and M. Marsili \at
	       Abdus Salam International Centre for Theoretical Physics, Strada Costiera 11, 34151 Trieste, Italy \\
	       \email{marco.bardoscia@ictp.it} \\
	       \email{glivan@ictp.it} \\
	       \email{marsili@ictp.it}
	       \and
	       G. De Luca \at
	       SISSA, International School for Advanced Studies, via Bonomea 265, 34136 Trieste, Italy \\
	       \email{giancarlo.deluca@sissa.it}
	       \and
	       C. J. Tessone \at
	       Chair of Systems Design ETH Zurich Kreuzplatz 5, CH-8032 Zurich \\
	       \email{tessonec@ethz.ch}
}

\date{Received: date / Accepted: date}

\maketitle

\begin{abstract}
The structure of societies depends, to some extent, on the incentives of the individuals they are composed of. We study a stylized model of this interplay, that suggests that the more individuals aim at climbing the social hierarchy, the more society's hierarchy gets strong. Such a dependence is sharp, in the sense that a persistent hierarchical order emerges abruptly when the preference for social status gets larger than a threshold. This phase transition has its origin in the fact that the presence of a well defined hierarchy allows agents to climb it, thus reinforcing it, whereas in a ``disordered'' society it is harder for agents to find out whom they should connect to in order to become more central. Interestingly, a social order emerges when agents strive harder to climb society and it results in a state of reduced social mobility, as a consequence of ergodicity breaking, where climbing is more difficult.
\end{abstract}

\section{Introduction} \label{sec:intro}

The emergence of social elites has interested social scientists ever since Pareto's observation of persistent inequalities in our societies~\cite{Pareto}.
Inequality is acceptable if it results from differences of individuals in terms of their capabilities, but not if it results, in one way or another, from discrimination\footnote{India's cast system or racial segregation in the US and South Africa in the last century, are examples of explicit discrimination of underprivileged groups, that in the course of time has come to be regarded more and more as unacceptable, prompting for  explicit measures of {\em affirmative action} (e.g.  quotas for lower casts in India).}.
Not only discrimination conflicts with ethical principles that all individuals are {\em a priori} equal and should have access to the same opportunities. It also damages societies in terms of efficiency \cite{ASen} as it hampers social mobility, preventing society from promoting individuals to positions in the social hierarchy that are consistent with their capabilities.

We introduce the {\em social climbing game}, a highly stylized model of a society, where individuals attempt to optimize their position in the network, by becoming as central as possible. The assumptions of the model are rooted on empirical and theoretical evidence coming from the social sciences.
There, in the early years of network analysis, it was found that the importance of an individual within a network is related to some quantification of
how {\em central} \cite{bavelas48,wasserman94} this agent is.
There exist different metrics which measure the centrality of a node (among others: degree, betweenness, closeness, eigenvector) \cite{freeman78}, each one highlighting different facets of this generic concept.

Among the empirical analyses, there is a body of literature showing that centrality explains the role, importance, or payoffs of the agents constituting the network: in informal structures within organizations, the importance of people is related to their betweenness centrality  \cite{brass84}; students with an higher centrality in the friendship network were found to perform better in education tests \cite{calvo-armengol09}\footnote{A more comprehensive list can be found in Ref.~\cite{koenig09}.}.
From the theoretical side, Ref.~\cite{ballester06} shows that in a broad class of games, player's payoffs increase with their (Bonacich) centrality  \cite{bonacich87} in the network.
Because of this, if individuals can alter their neighborhood, the myopic best response strategy is simply to connect to the neighbor who increases their centrality the most.


Interestingly, K\"oenig {\em et al.}~\cite{koenig09} have shown that when individuals strive to be as central as possible, the exact measure of centrality is irrelevant, and the dynamics yields a network which has the property of {\em nestedness}: the neighborhood of any node contains the neighborhood of the nodes which have a lower degree. In this kind of networks, the ranking of nodes according to their centrality is the same, regardless of the centrality measure considered \cite{koenig11}. 
Remarkably, nested structures have been found in inter-organizational networks of research and development (R\&D) alliances \cite{Uzzi1996,Saavedra2008}, in interbank payment networks \cite{Soramaki2007} and in firm competition under oligopolies \cite{goyal03}. This kind of structures will be precisely the ones emerging in the social climbing game. In this respect, our results confer stability to those of Ref.~\cite{koenig11} and generalize them in non-trivial ways.

This study suggests that the assumption that individual freedom promotes social mobility is a non-trivial one. This is because the structure of a society, while constraining the set of opportunities that are available to individuals, depends on the very incentives of individuals in complex ways.
In this paper we show that this interplay may produce very ``rigid'' societies, with extremely low social mobility, characterized by persistent inequalities between {\em a priori} equal individuals\footnote{The positive relationship between intergenerational social mobility and inequality has been consistently reported in several empirical studies  \cite{WilkinsonPickett2010,BjorklundJantti1997,AndrewsLeigh2009}.}. The understanding of this phenomenon hinges on the concept of ergodicity breaking that occurs in strongly interacting systems, when a symmetry -- here related to the {\em a priori} equality among individuals -- is spontaneously broken. This phenomenon, well known in statistical physics, is an emergent collective property, and it manifests only when the system is large enough.
Remarkably, we find that persistent inequality with low mobility occurs precisely when the quest for ``power'' -- i.e. for occupying the most central or important place in the social hierarchy -- becomes a dominant component of what motivates the behavior of individuals.

In words, our model epitomizes an apparent positive feedback between the intensity of the efforts of individuals to ``climb'' the social hierarchy and the structure of a society: on the one hand, the more a society is hierarchically structured, the easier it is for individuals to understand how to climb it. On the other, the efforts of agents to climb the hierarchy reinforce the social ranking as individuals rewire their links from less to more influential individuals. We discuss this interplay in a highly stylized model of a society, that while being very far from realistic, serves as a proof of concept and allows us to unveil the mechanism responsible for the emergence of a persistent inequality in a transparent manner.

In addition, this approach shows the relevance of 
techniques used in statistical mechanics \cite{Park2004a,Park2004b} in the context of social networks. Similar models to the one considered here have been discussed in Refs.~\cite{Derenyi2004583,lambdamax}) that, however, focus mostly on topological properties of the emerging networks.

The rest of the paper is organized as follows: In Sect.\ \ref{sec:model}, we introduce the model and discuss the main properties of the dynamics and its associated global potential function; related to these results, in Appendix \ref{sec:appendix}, the ergodicity of the system is proved. Later, in Sect.\ \ref{sec:results} we show the results of extensive numerical simulations that portray the characteristic behavior of the system. Finally, in Sect.\ \ref{sec:conclusions}, the conclusions are drawn.

\section{The model} \label{sec:model}

We consider a system composed of $N$ individuals, who are connected through a network which consists of exactly $M$ links. The network is undirected and thus can be specified in terms of a symmetric adjacency matrix $\hat a=\{a_{ij}\}_{i,j=1}^N$, with elements $a_{ij}=a_{ji}=1$, if $i$ and $j$ are connected, $a_{ij}=a_{ji}=0$ otherwise. Agents receive opportunities to use their links in order to get in contact with more ``influential'' members of the society, in brief to climb the social network.

As a measure of importance of the individuals, we take the number of his/her partners\footnote{Other measures of centrality can be taken but, as observed in Ref.~\cite{koenig11}, these rank individuals in the same order in strongly hierarchical networks, that will be stable over time as we shall see later. Conversely, unstructured networks correspond to random rankings with no stable order, with respect to all centrality measures.}, i.e.\ the degree $k_i = \sum_j a_{ij}$. As a measure of the ``social capital'' of agent $i$ we take  the following \emph{local} utility function
\begin{equation} \label{eq:utility}
u_i = \sum_{j,\ell=1}^N a_{ij} a_{j\ell} + \mu \sum_{\ell=1}^N a_{i\ell}= \sum_{j=1}^N a_{ij}k_j + \mu \, k_i,
\end{equation}
that depends both on the centrality $k_i$ of agent $i$ and on the centrality $k_j$ of his/her neighbors, with $\mu$ tuning the relative weight between the two terms\footnote{As will be clear in the following, the second term in (\ref{eq:utility}) is irrelevant for the dynamics, but not for the interpretation of the local utility. For example, consider the limit case of a star: while the central node is connected to $N-1$ nodes, all other nodes have only one connection. In this case the first term in \eqref{eq:utility} is equal to $N-1$ for all nodes and only the term proportional to $k_i$ removes this degeneracy. Note that the second term in \eqref{eq:utility} also describes a linear cost $\mu < 0$ to maintain  links.}
The efforts of agents to climb the social hierarchy can then be formalized in the maximization of the utility $u_i$.

We then define the dynamics as follows,
\begin{enumerate}
  \item At any time, an agent $i$ is picked at random together with one of her neighbors, $j$. Then, a neighbor $\ell$ of $j$ is selected at random, $\ell \neq i$.
  \item If $\ell$ is already connected to $i$, nothing happens. Otherwise, with probability
\begin{equation} \label{eq:prob}
	p_{(i,j)\to (i,\ell)}=\frac{e^{\beta \Delta u_i}}{1+e^{\beta \Delta u_i}},
\end{equation}
the link $(i,j)$ is replaced with (or rewired to) link $(i,\ell)$, 
where $\Delta u_i$ is the corresponding change in $i$'s utility.
\end{enumerate}
The step 1 models random encounters between agents through their network of interactions. In such an encounter, agent $i$ gets to know a friend $\ell$ of $j$, as well as his/her importance (the number $k_\ell$ of $\ell$'s friends). The probabilistic choice rule in step 2 can be derived from a random utility model where agents maximize a more complex utility function, that accounts for the fact that the social network affects in complex ways the well being of individuals and their unobserved choices in other dimensions\footnote{This idea can be precisely formalized assuming that $u_i(\hat a)$ is the observed part of the utility, but that agent $i$ maximize a more complex function $U_i(\hat a,\vec b)=u_i(\hat a)+v_i(\vec b| \hat a)$ where $v_i(\vec b| \hat a)$ is a random unobserved contribution to the utility, that depends on a vector $\vec b$ of unobserved choices. Assuming that $v_i(\vec b| \hat a)$ are independent and identically distributed, it can be shown (See \cite{ChalletMarsiliZhang} p.~33 for an explicit derivation) that $\max_{\vec b} U_i(\hat a,\vec b)=u_i(\hat a) + \eta_i(\hat a) / \beta$, where $\eta_i(\hat a)$ are i.i.d. with a Gumbel distribution. It is well known \cite{mcfadden} that if $\hat a^*$ is the choice that maximizes $u_i(\hat a) + \eta_i(\hat a) / \beta$, then $P\{\hat a^*=\hat a\}$ is given by Eq. \eqref{eq:prob}.}. 
In this view, $\beta$ plays the role of the relative weight between the observed and the unobserved part of the utility in the the choice of social contacts and it reflects the prevalence of the quest for social status in their choice behavior\footnote{For example, Adam may be reluctant to interrupt his relation with Bob, despite his low rank in society, because he is his only friend who shares his interest in Japanese paintings.}. In particular, in the limit $\beta \rightarrow \infty$, a move implying a decrease in the utility function is never accepted. This means that the social status is valued so highly by the agents that everything else is unimportant. On the contrary, for $\beta = 0$ the probability of accepting a move implying a decrease of the utility function is $1/2$, meaning that the social status has negligible importance with respect to the unobserved part of the utility. The general question addressed is then how strong should the parameter $\beta$ be in order for a social hierarchy to form and be maintained in the long run?

It is worth to remark that if the utility of agent $i$ increases when rewiring the link $(i,j)$ to $(i,\ell)$, then the utility of agent $j$ decreases, while that of agent $\ell$ increases. This embodies the fact that the formation of a new link needs the consent of both parties, but their removal can be unilateral. Therefore, we can interpret the rewiring mechanism as a process according to which agent $i$ looks for some social premium (e.g. knowledge of information, professional expertise) that agent $\ell$ can provide more than agent $j$. Once agent $i$ secures his/her connection to agent $\ell$, agent $j$ essentially represents a redundant, less central source of the same capital, and this is why the rewiring operation happens at his/her expenses.
Moreover, the rewiring mechanism described above implies that, in their quest to become central, agents increase the likelihood to be selected by others as new partners.

Notice finally that the number of links is conserved in the dynamics. Hence the density of links is the second important dimension that we shall explore, in order to understand how the structure of social organization depends on it.

\subsection{Properties of the dynamics and potential function}

There are some remarkable features of the dynamics of the model introduced in the previous sections. We detail them now.

First, it is easy to see that the dynamics introduced preserves connected components. Indeed,  nodes are never disconnected by the dynamics because, even if they have just one link, this will not be rewired because the neighbor upstream has no second neighbor where to rewire.
Therefore, without loss of generality, we restrict attention to the case where $M \geq N-1$ and the network is composed of a single connected component. Networks composed of disjoint components remain disjoint under the dynamics above, hence the dynamics of different components can be considered independently. Alternative dynamics that do not preserve connectedness -- e.g. adding the link $(i,\ell)$ to a neighbor $\ell$ of a neighbor $j$, and removing a link different from $(i,j)$ chosen in any way -- would converge to simple structures characterized by cliques of $\sim \sqrt{M}$ nodes in a sea of disconnected nodes. Indeed, it is easy to check that such configurations correspond to absorbing states of the dynamics for all $\beta$. On the other hand, as we shall discuss in a moment, it is precisely the rewiring procedure we propose in Section \ref{sec:model} that produces non-trivial equilibrium states.

Notice that, since both the number $N$ of nodes and $M$ of edges is conserved during the evolution of the system, the number of fundamental cycles in the graph is also conserved. This follows from the fact that the number  of fundamental cycles in a graph is  equal to $M - N + K$, where $K$ is the number of connected components (see ~\cite[Ch.\ 2]{Bollobas1998}).

The dynamics of the model admits a potential which is just the \emph{global} utility, i.e.~the sum of the utilities $U = \sum_i u_i$.
Indeed, let us consider the change $\Delta u_x$ in the utility of the agent $x$ when the rewiring $(i,j)$ into $(i,\ell)$ occurs. Depending on the position of $x$ in the network, the following changes are obtained:
\begin{subequations}
\begin{align} \label{eq:delta_local_utility}
	\Delta u_i & = k_\ell - k_j + 1 \\
	\Delta u_j & = 1 - k_i - \mu \\
	\Delta u_\ell & = k_i - 1 + \mu\\
	\Delta u_h & = -1 \qquad \forall h \in \partial j \setminus{\{i, \ell\}} \\
	\Delta u_g & = +1 \qquad \forall g \in \partial \ell \setminus{\{j\}} \\
	\Delta u_x & = +0 \qquad \forall x \neq i, j, \ell, x \notin \partial j\cup \partial \ell \, ,
\end{align}
\end{subequations}
where $\partial x$ is the set of the neighbors of $x$, before the move.

In the total variation of the utility $\Delta U=\sum_x\Delta u_x$, the term $\Delta u_h$ appears $k_j - 2$ times, while the term $\Delta u_g$ appears $k_\ell - 1$ times, because $k_x$ is the degree of the node $x$ \emph{before} the rewiring. Gathering all the contributions one has:
\begin{equation} \label{eq:delta_utility}
\begin{split}
	\Delta U & = \Delta u_i + \Delta u_j + \Delta u_\ell + (k_j - 2) \Delta u_h + (k_\ell - 1) \Delta u_g \\
	         & = 2 (k_\ell - k_j + 1) = 2 \Delta u_i \, .
\end{split}
\end{equation}
The last point implies that, provided the dynamics is ergodic, which is proven in Appendix  \ref{sec:appendix}, the system converges to thermal equilibrium with Hamiltonian
$$
\mathcal{H} = -U = -\sum_i k_i^2 - \mu \sum_i k_i,
$$
and fixed density of links at temperature $2/\beta$\footnote{The factor 2 comes from the fact that the variation of the global utility is the double of the variation of the local utility.}. Notice that the second term does not play any role, being $\sum_i k_i$ a fixed quantity in our case. Indeed the dynamics in Eq.~(\ref{eq:prob}) is equivalent to Metropolis dynamics, and hence it 
samples the Gibbs distribution $P\{\hat a\}\propto e^{\beta U(\hat a)/2}$, which is known in sociology as the 2-star model. 
Park and Newman~\cite{Park2004a,Park2004b}, have shown that the 2-star model where the density of links is not fixed, exhibits a sharp phase transition. This result suggests that there might be a phase transition also in the model we study in this paper. As a byproduct, our discussion also provides a microeconomic derivation for the 2-star model\footnote{The case studied in ref.~\cite{Park2004a,Park2004b} where the number of links is also allowed to change, can be recovered in a model where, in addition to rewiring steps discussed above, we also allow for link creation upon random encounters and link obsolescence (i.e. decay). More precisely, 
consider a model where each agent receives opportunities {\em i)} to rewire his/her links (as above) at rate $\nu$ and {\em ii)} to form new links (with randomly chosen agents), with rate $\eta/2$. In addition, each link decays with rate $1$. Then, in a time interval $\Delta  t$, the number of links changes by $\Delta M=\eta N\Delta t-M\Delta t$, which means that in the stationary state $\langle M\rangle=\eta N$.}.

\section{Numerical simulations} \label{sec:results}

In order to investigate the behavior of the model, we performed extensive numerical simulations sampling the Gibbs distribution $P\{\hat a\}\propto e^{\beta U(\hat a)/2}$ using the Metropolis algorithm based on the rewiring moves introduced in Sect.\ \ref{sec:model}. All the results to be presented throughout the rest of this section were obtained, for each value of $\beta$, by performing $R $ rewiring proposals per node, and we checked that the value $R = 5 \cdot 10^5$ is large enough to always ensure the attainment of an equilibrium state. Fig.\ \ref{fig:nets} shows two typical realizations of the social network for small and large values of $\beta$ (see caption for more details). Fig.\ \ref{fig:nets} suggests that, as anticipated in the previous section, the social climbing model undergoes a transition from hierarchical to random structures. In the following, we will show the presence of a phase transition between these two states.

\begin{figure}
	\centering
  	\includegraphics[width=0.48\columnwidth]{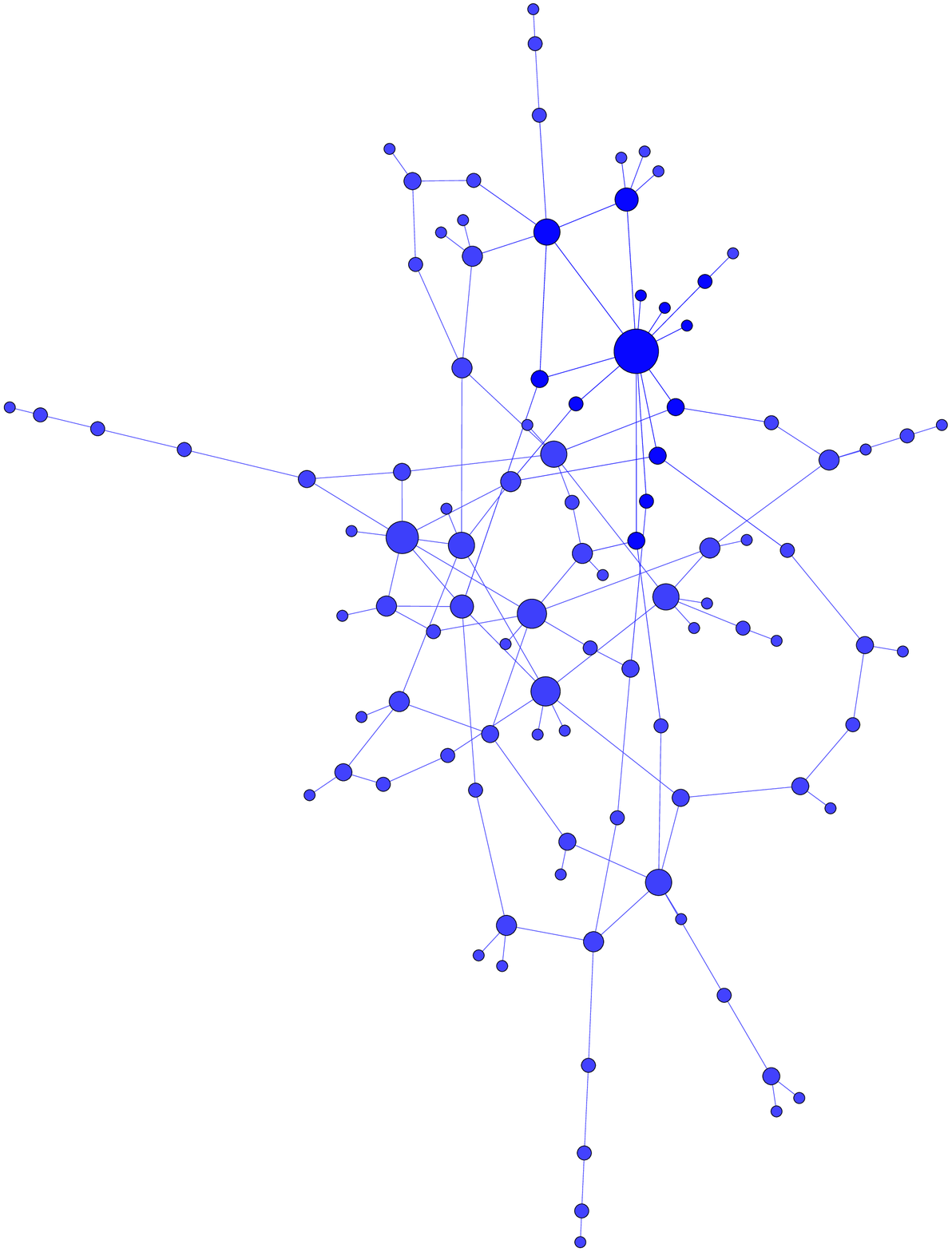}
  	\includegraphics[width=0.48\columnwidth]{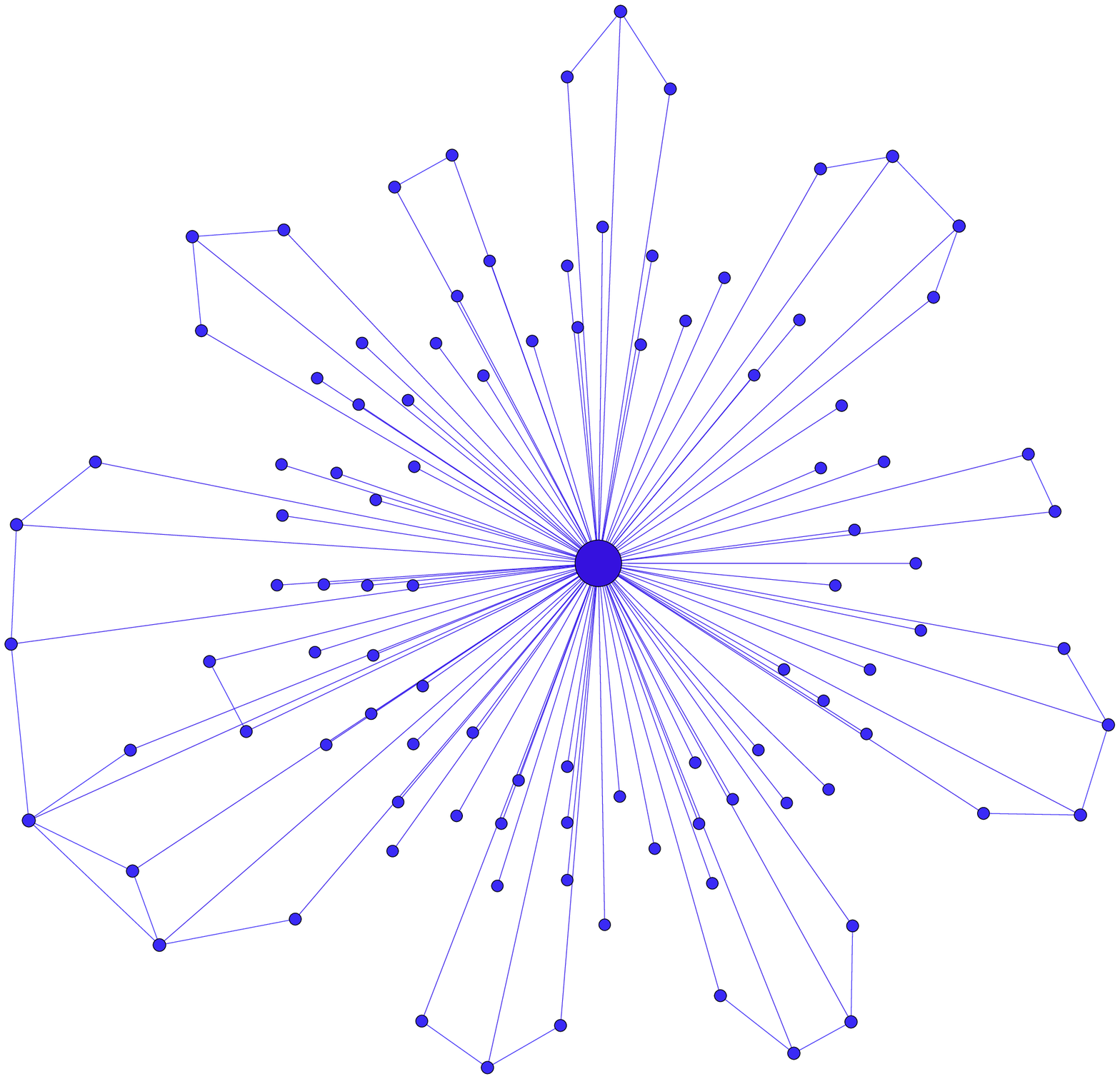}	
\caption{(Color online) Snapshot of networks of the social climbing game for $N = 100$, $M = 125$ for $\beta = 0.03$ (left panel) and $\beta = 0.1$ (right panel). Size of the nodes is proportional to the degree.}
\label{fig:nets}
\end{figure}

\begin{figure}
	\centering
  	\includegraphics[width=0.48\columnwidth]{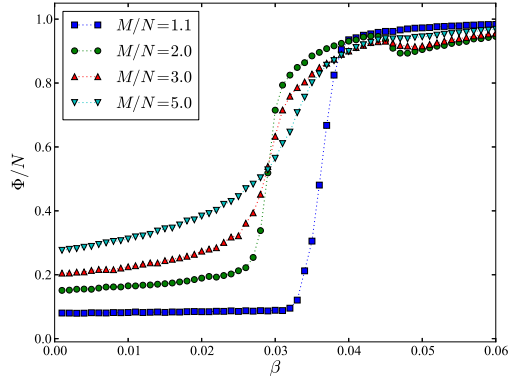}
\caption{(Color online) Dependence of the largest degree $\Phi$ (divided by $N$) in the social climbing network as a function of the inverse temperature (or intensity of choice parameter) $\beta$. The different curves refer to $N=100$ and $M=110$, $200$, $300$, $500$. For each value of $\beta$ the reported values of $\Phi$ are obtained by averaging over 100 networks. An abrupt change in is observed in all curves after a threshold value of $\beta$, with $\Phi / N$ going from low values to values close to one, signaling the emergence of a star, i.e. a link with $\mathcal{O}(N)$ links, in the network.}
\label{fig:max_deg}
\end{figure}

In Fig.\ \ref{fig:max_deg} we show the largest degree of the network $\Phi = \max_i k_i$ as a function of the inverse temperature $\beta$ for systems with $N = 100$ nodes and $M = 110$, $200$, $300$, $500$ links. As can be seen, in all cases the system actually undergoes a transition, going from a phase where the largest degree $\Phi$ is roughly of order $1-10$ (depending on the relative size of $N$ and $M$) to a phase where the largest degree is of order $N$. These observations qualitatively match the findings of \cite{Derenyi2004583}, where a prediction for the critical temperature $T_c = 1/\beta_c$ characterizing this phase transition was also derived, from combinatorial arguments, for networks with average degree $\bar{k} = 2M/N <2$, i.e. for disconnected graphs. The nature of the phase transition depicted in Fig.\ \ref{fig:max_deg} is further investigated in Fig.\ \ref{fig:transition}. In the left panel we show the full distribution of $\Phi / N$ with respect to $\bar{k}$ obtained by binning the results relative to 100 networks, for $\beta=0.01$ and $N=500$. For low (high) values of $\bar{k}$ the distribution is sharply concentrated around zero (one) and a steep transition occurs at a critical value of $\bar{k}$, meaning that the average is representative of the distribution of $\Phi / N$. Completely analogous results are found for different values of $\beta$. Therefore, in order to characterize the transition more precisely, in the right panel of Fig.\ \ref{fig:transition} we show the relation between the average of $\Phi / N$ over 100 networks with respect to both $\bar{k}$ and $\beta$.
\begin{figure}
	\centering
  	\includegraphics[width=0.48\columnwidth]{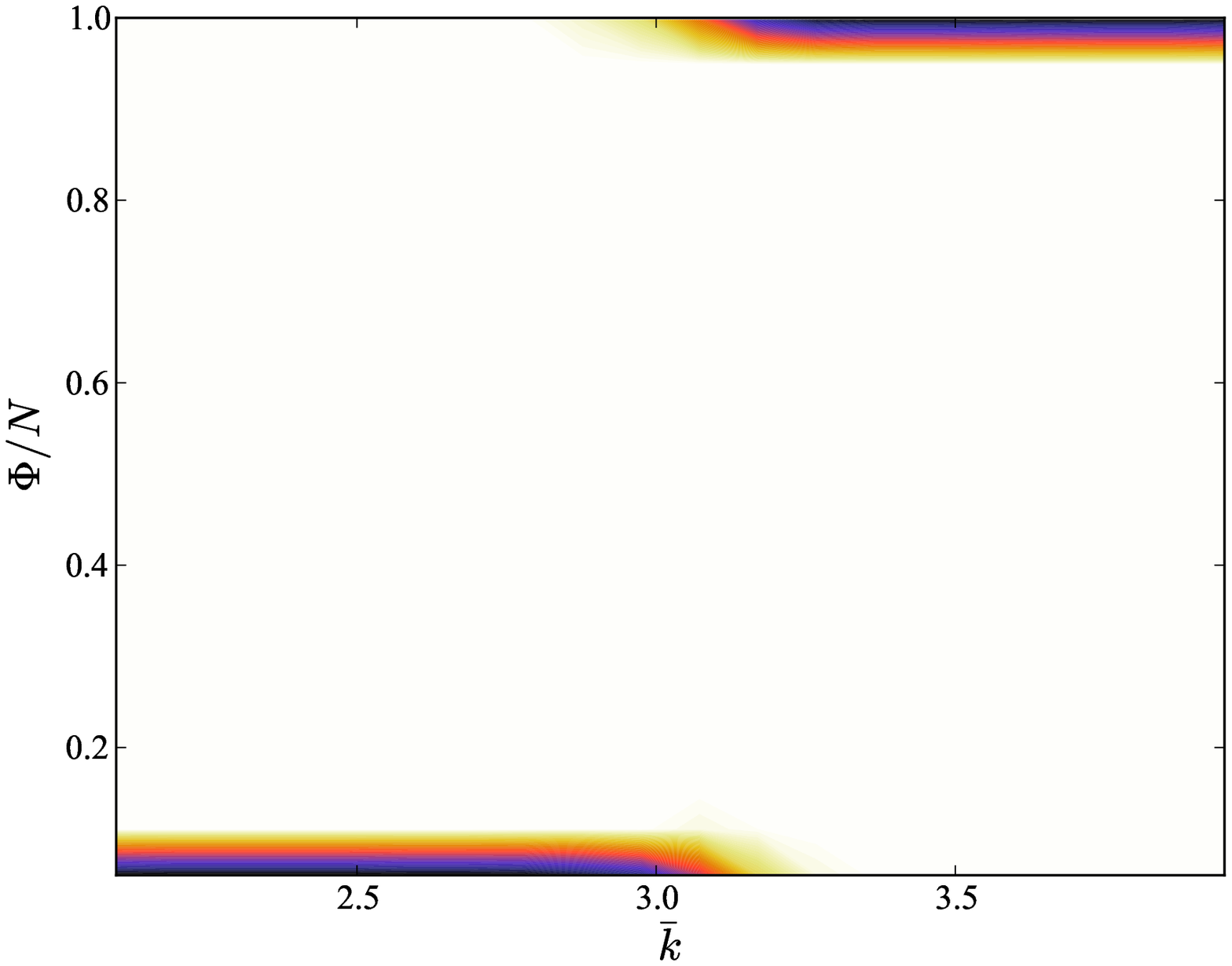}
  	\includegraphics[width=0.51\columnwidth]{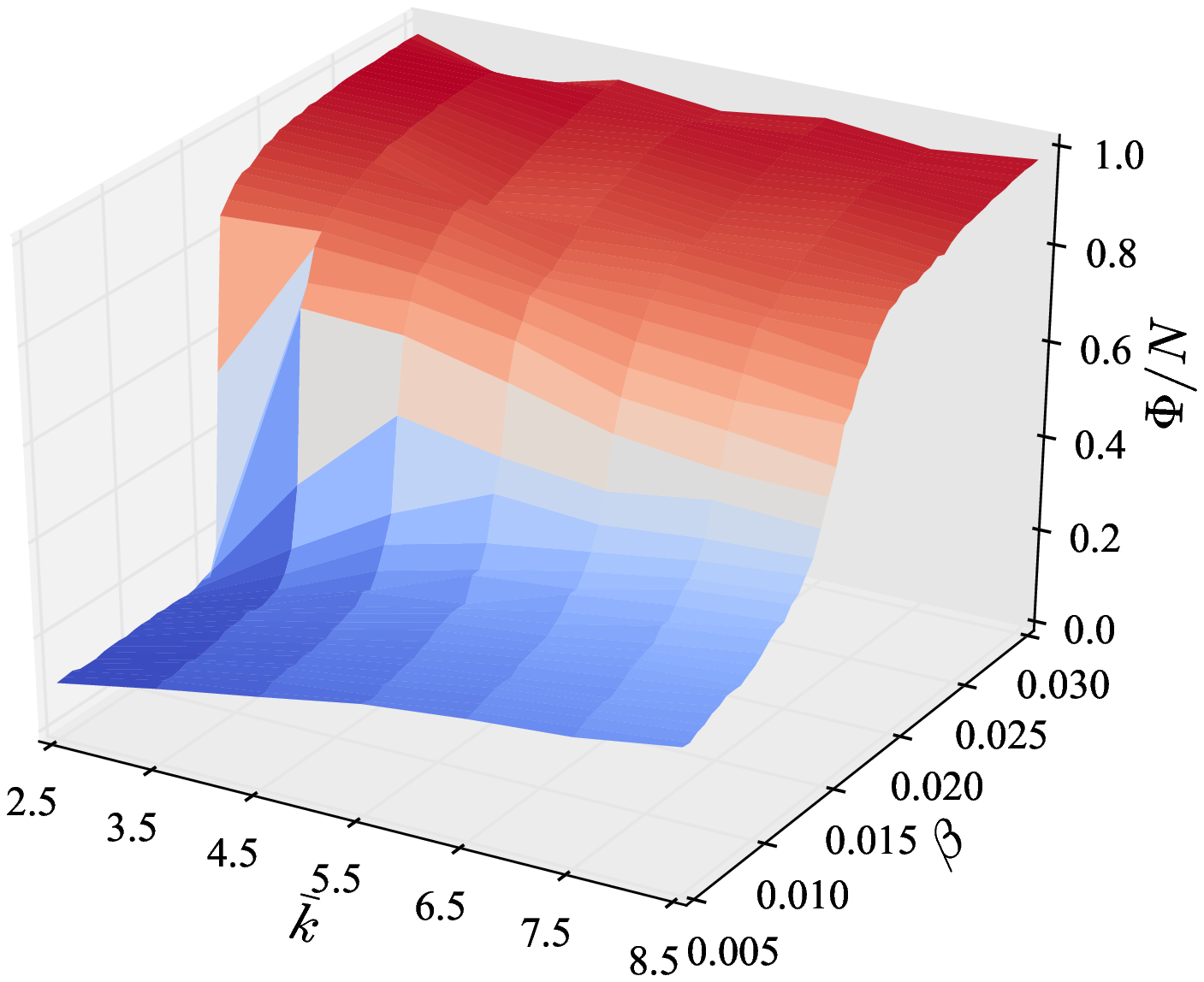}
\caption{(Color online) Left panel: density plot depicting the full distribution of the maximum degree (divided by $N$) as a function of the average degree obtained by binning the results relative to 100 networks, at inverse temperature $\beta = 0.01$ and number of nodes $N=500$. High (low) values are darker (lighter). Right panel: average maximum degree (divided by $N$) as a function of the average degree and $\beta$ for $N=100$; results obtained by averaging over 100 networks. Clearly, for low values of $\bar{k}$ the network is in the disordered phase, while for high values of $\beta$ it is in the ordered phase.}
\label{fig:transition}
\end{figure}

In Fig.\ \ref{fig:scaling} we analyze the dependence of the critical value of $\beta$ with respect to the size $N$ of the network, while keeping the average degree fixed, for $\bar{k} = 2.5, \, 5.0$. Qualitatively it is clear that, increasing $N$, both the transition becomes sharper and the critical value of $\beta$ shifts to the left. In order to understand if the critical value in thermodynamic limit $\beta_c$ is nonzero, we analyze the finite-size scaling behavior, assuming $\beta^*(N) = \beta_c + a N^{-b}$, where $\beta^*(N)$ is the critical value at size $N$ and $\beta_c$, $a$ and $b$ are free parameters. Since $b$ is expected to be universal (i.e.\ not dependent on the other parameters, like $\bar{k}$), it is reasonable to choose it by plotting $\beta^*(N)$ against $N^{-b}$ until straight lines are obtained. Both $a$ and $\beta_c$ are then found by a best fit. The value of $\beta^*(N)$ is obtained by a linear interpolation of the curves in Fig.\ \ref{fig:scaling} and calculating the value of $\beta$ such that $\Phi / N = 1/2$. From Fig.\ \ref{fig:scaling_joined} it can be clearly seen that for $b = 1.25$ the assumed functional form is fully consistent with numerical simulations up the investigated system size. The values we find for $\beta_c$ are soundly different from zero within $95\%$ confidence intervals provided by best fit. In particular we find for $\bar{k} = 2.5$: $\beta_c = (1.1 \pm 0.2) \cdot 10^{-2}$, while for $\bar{k} = 5.0$: $\beta_c = (2.6 \pm 0.2) \cdot 10^{-3}$.\footnote{Inspired by \cite{Derenyi2004583} we also performed  finite-size scaling analysis according to the functional form: $\beta^*(N) = \beta_c + a (M/\log(N))^{-b}$, which also gives values of $\beta_c$ soundly above zero and consistent with the ones discussed in the main text.} 
\begin{figure}
	\centering
  	\includegraphics[width=0.48\columnwidth]{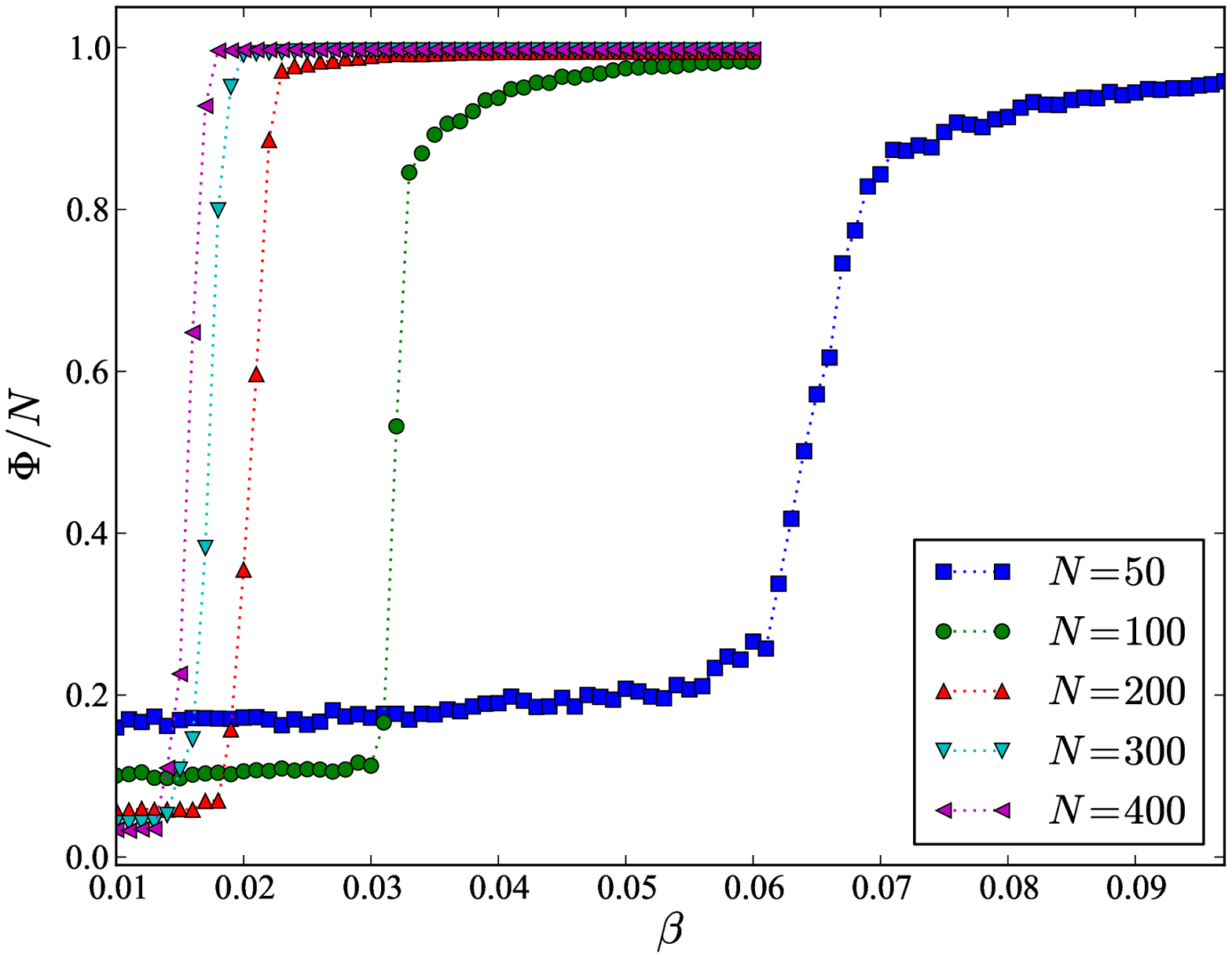}
  	\includegraphics[width=0.48\columnwidth]{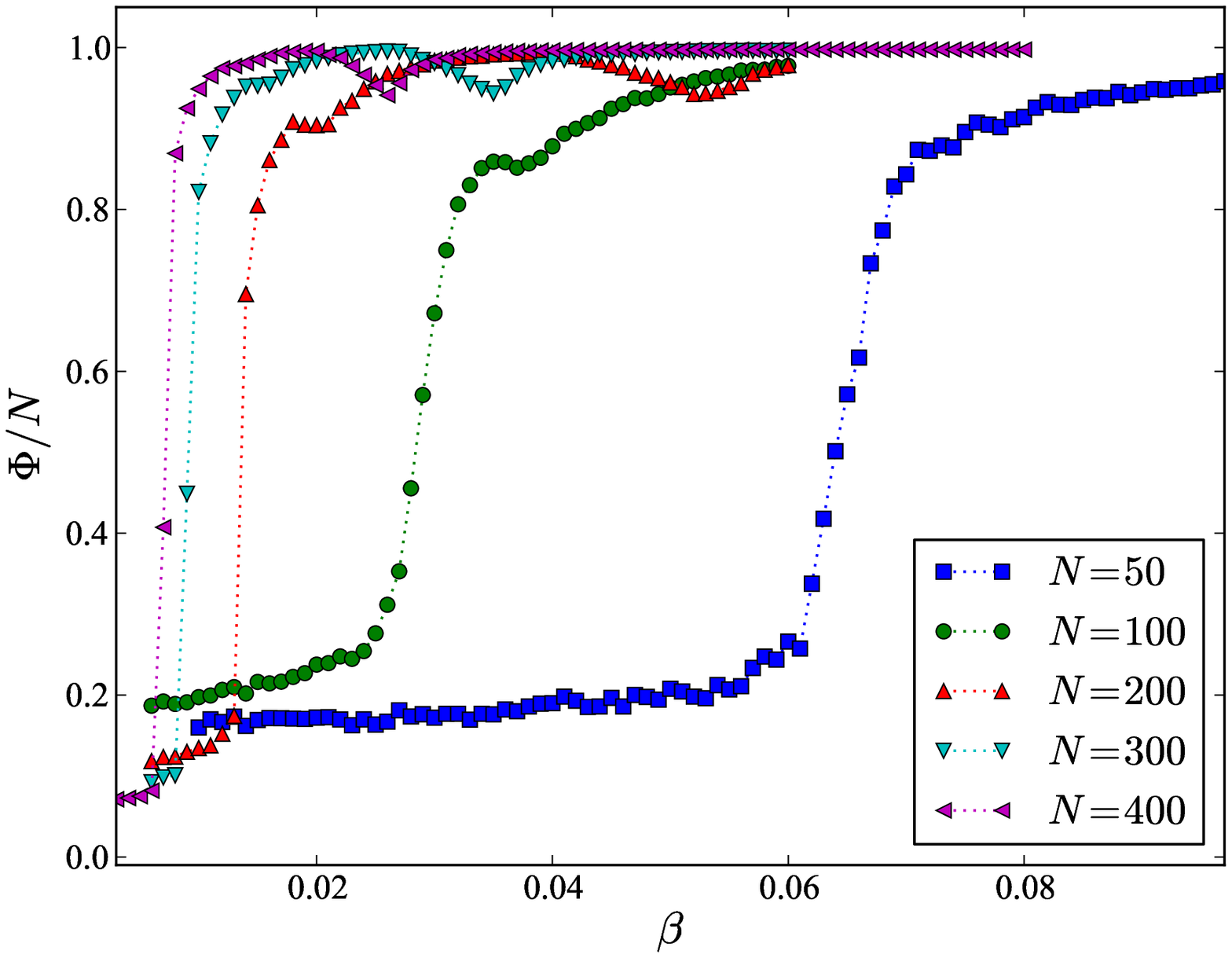}
\caption{(Color online) Left panel: Dependence between the maximum degree (divided by $N$) and $\beta$, for $\bar{k} = 2.5$ and different values of $N$. Results are averaged over 100 networks. The transition become sharper and the critical value of $\beta$ shifts to the left for increasing values of $N$. Right panel: as in the left panel for $\bar{k} = 5.0$.}
\label{fig:scaling}
\end{figure}

\begin{figure}
	\centering
  	\includegraphics[width=0.48\columnwidth]{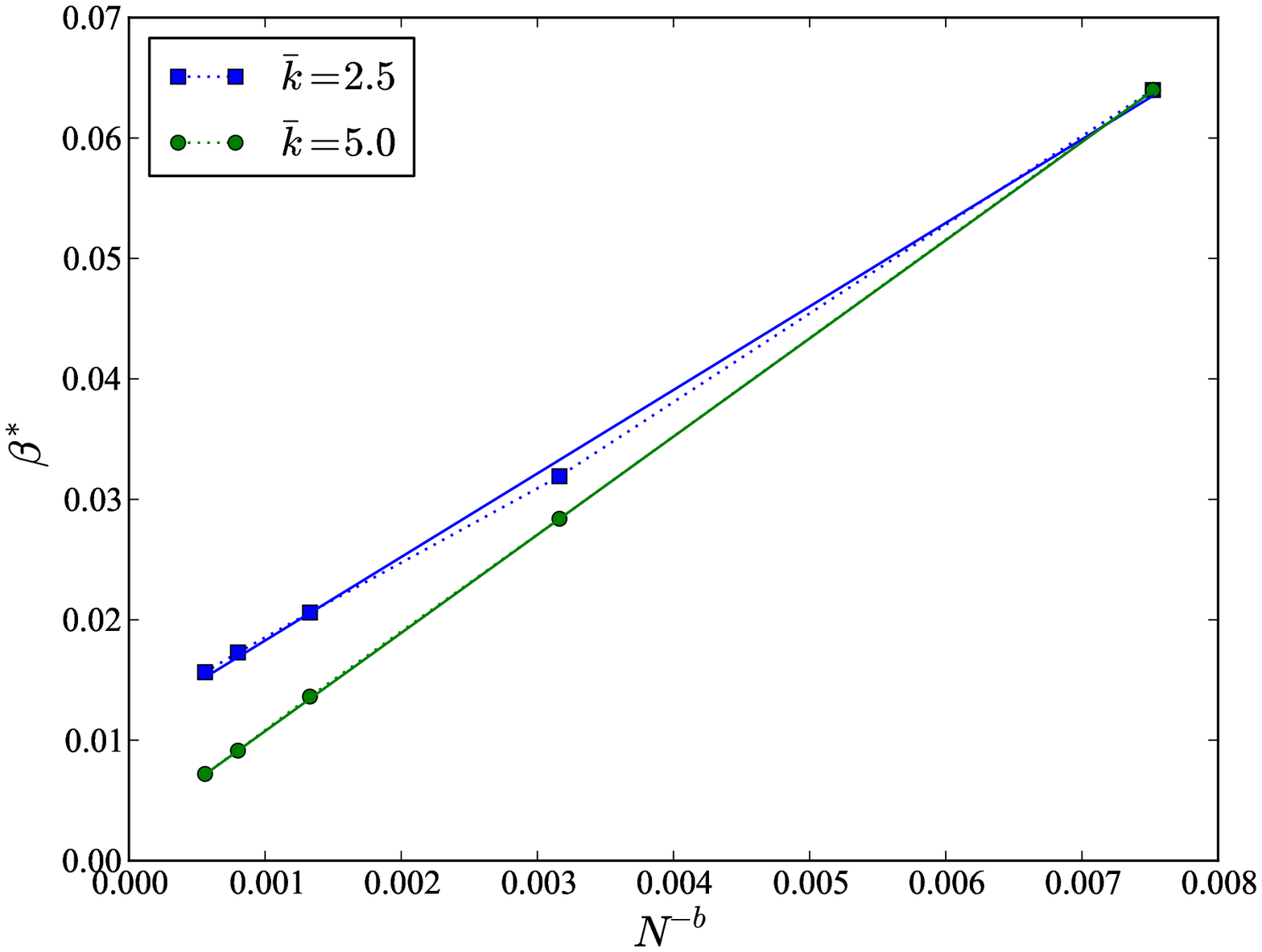}
\caption{(Color online) Finite-size scaling for determining the critical value in the thermodynamic limit $\beta_c$, assuming the functional form $\beta^*(N) = \beta_c + a N^{-1.25}$, where $\beta^*(N)$ is the critical value at size $N$. By best fitting we find: for $\bar{k} = 2.5$: $\beta_c = (1.1 \pm 0.2) \cdot 10^{-2}$, while for $\bar{k} = 5.0$: $\beta_c = (2.6 \pm 0.2) \cdot 10^{-3}$.}
\label{fig:scaling_joined}
\end{figure}

Following \cite{Derenyi2004583}, let us define a star as a node whose degree is of order $N$. Then, it is immediate to figure out that, depending on the ratio $M/N$, different number of stars might emerge in the network for temperatures lower than $T_c$. Clearly, in the case $N=100$, $M = 110$ (i.e.\ $\bar{k}$ only slightly larger than 2), the appearance of a star ($\Phi \simeq 100$ in this case) below the critical temperature leaves very few links to be distributed amongst the remaining nodes. On the other hand, increasing the number of links provides enough room for the emergence of a larger number of stars. In other words, it is intuitively reasonable to expect a system with an average degree $\bar{k} \simeq 2n$ to produce, for sufficiently low temperatures, exactly $n$ stars. In order to support such an intuitive line of reasoning, we computed the inverse participation ratios (IPRs) of the degree sequences $\mathbf{v} = (k_1, k_2, \ldots, k_N) / \sqrt{\sum_{i=1}^N k_i^2}$ of several networks with different numbers of nodes and links. Given a normalized vector $\mathbf{v}$, its IPR is defined as
\begin{equation} \label{eq:IPR}
I(\mathbf{v}) = \left ( \sum_{i=1}^N v_i^4 \right )^{-1} \, .
\end{equation}
The IPR of a completely localized vector, say $\mathbf{v} = (1,0,\ldots,0)$, is equal to one. On the other hand, the IPR of a fully delocalized vector, whose components are all equal to $v_i \simeq 1/\sqrt{N}$, is of order $N$. In our case $I(\mathbf{v})$ gives an estimate of the number of dominant nodes in the network. In Fig.\ \ref{fig:ipr} we plot the IPRs $I(\mathbf{v})$, as functions of $\beta$, for $N = 100$ nodes and $M = 110$, $200$, $300$. As can be seen, the IPR of the sparsest network, i.e.\ the one with $M=110$, essentially drops down to one right below its critical temperature. On the other hand, systems with a larger number of links undergo a less trivial evolution: after the initial drop below the critical temperature, the IPR increases and eventually reaches a steady state. In the example shown in Fig.\ \ref{fig:ipr}, the system with $M=200$ links reaches a steady value $I(\mathbf{v}) \simeq 2.12 \pm 0.03$, whereas the system with $M=300$ reaches $I(\mathbf{v}) \simeq 3.28 \pm 0.06$ (where the errors represent the $68\%$ confidence intervals obtained by averaging over 100 networks), and such values clearly show that the maximal number of stars allowed by the relative sizes of $N$ and $M$ has been achieved. Moreover, these observations are consistent with the small temporary decrease of the largest degree $\Phi$ which can be observed in Fig.\ \ref{fig:max_deg} for systems with $\bar{k} > 2$ when the inverse temperature is slightly larger than its critical value.

\begin{figure}
	\centering
  	\includegraphics[width=0.48\columnwidth]{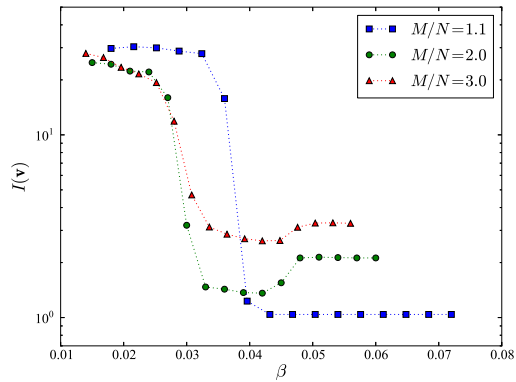}
\caption{(Color online) Inverse participation ratio of the normalized degree sequence $I(\mathrm{\vec{v}})$ as a function of the inverse temperature $\beta$. Different curves refer to networks with $N=100$ and $M=110$, $200$, $300$. Results obtained by averaging over 100 networks. For large enough $\beta$ one recovers the maximal number of stars.}
\label{fig:ipr}
\end{figure}

\subsection{Correlations and social mobility}

As already explained in Sect.\ \ref{sec:intro}, one of the goals of the present paper is to model the positive feedback mechanism between the individuals' effort to climb the social hierarchy and the subsequent reinforcement of the social hierarchy itself. Suppose that a given social network reaches its equilibrium state, at a certain inverse temperature $\beta$, after $t_0$ steps of the social climbing dynamics described in Sect.\ \ref{sec:model}. Let us denote the corresponding graph's adjacency matrix as $\hat{a}(t_0)$. Then, one way of quantitatively describing how mobile or ``frozen'' a society is would be to assess the level of correlation, according to some proper notion, between $\hat{a}(t_0)$ and a following configurations $\hat{a}(t)$, where $t = t_0 + \Delta t$ for some positive $\Delta t$. We will now measure correlations by means of Kendall's rank correlation coefficient. Given the joint set of all matrix entries in $\hat{a}(t_0)$ and $\hat{a}(t)$, let us focus, for example, on entries $(i,j)$ and ($h,\ell$) in both matrices. Then, if both $a_{ij}(t_0) > a_{h\ell}(t_0)$ and $a_{ij}(t) > a_{h\ell}(t)$, or if both $a_{ij}(t_0) < a_{h\ell}(t_0)$ and $a_{ij}(t) < a_{h\ell}(t)$, the pairs $(a_{ij}(t_0),a_{h\ell}(t_0))$ and $(a_{ij}(t),a_{h\ell}(t))$ are said to be concordant. On the contrary, if $a_{ij}(t_0) \gtrless  a_{h\ell}(t_0)$ and $a_{ij}(t) \lessgtr  a_{h\ell}(t)$ the pairs $(a_{ij}(t_0),a_{h\ell}(t_0))$ and $(a_{ij}(t),a_{h\ell}(t))$ are said to be discordant. Of course, since the adjacency matrix entries equal zero or one at each time, ties will often happen either at time $t_0$ or at time $t$ (or at both times). Kendall's correlation coefficient $\tau$ reads

\begin{equation} \label{eq:Kendall}
\tau(\Delta t) = \frac{C - D}{\sqrt{C+D+T_{t_0}}\sqrt{C+D+T_{t}}},
\end{equation}
where $C$ ($D$) is the numbers of concordant (discordant) pairs, whereas $T_{t_0}$ ($T_{t}$) denotes the number of time-$t_0$ (time-$t$) ties. Pairs where ties happen both at  $t_0$ and $t$ are not taken into account.
\begin{figure}
	\centering
  	\includegraphics[width=0.48\columnwidth]{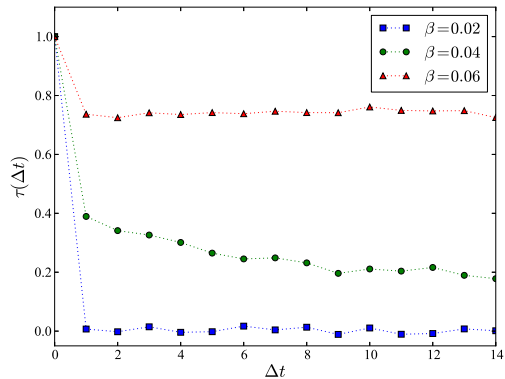}
\caption{(Color online) Kendall's $\tau$ coefficient (see \eqref{eq:Kendall}) measurements for networks with $N=100$ and $M=300$. All measurements are performed between an initial equilibrium configuration $\hat{a}(t_0)$ and later configurations $\hat{a}(t_0 + n\Delta t)$, with $t_0 = \Delta t = 5 \cdot 10^6$ Monte Carlo steps. The different curves refer to inverse temperatures $\beta = 0.02$, $0.04$, $0.06$, respectively corresponding to values below, slightly above and well above the critical value for the system under study (see also Fig.\ \ref{fig:max_deg}). Results obtained by averaging over 100 networks.}
\label{fig:Kendall}
\end{figure}

In Fig.\ \ref{fig:Kendall} a few examples of Kendall's $\tau$ coefficient's time evolution are sketched. All plots refer to networks with $N=100$ nodes and $M=300$ links. Here, $\Delta t = t_0 = 5 \cdot 10^6$ elementary Monte Carlo moves, i.e.\ rewiring proposals. As can be seen, when the social climbing game takes place for temperatures higher than the critical one, Kendall's $\tau$ quickly starts to fluctuate around zero, denoting no genuine correlation between configurations distant (in time). This we take as indication of a large social mobility. On the other hand, for temperatures slightly lower than the critical one, Kendall's $\tau$ remains significantly larger than zero over several time lags. However, a downward trend is clearly visible in this case, meaning that for temperatures $T \lesssim T_c$ social mobility is recovered after a sufficiently long time. On the contrary, for temperatures significantly lower than the critical one Kendall's $\tau$ essentially remains constant and very large (i.e.\ close to one) over large time lags, hinting at an extremely reduced social mobility, possibly preventing the majority of individuals from climbing the social ladder.

The above considerations on individuals' mobility in the social climbing game can be further clarified and understood more deeply. For these purposes, let us denote as $q_i$ the fraction of agents who, at a given time, have a strictly lower degree than agent $i$, i.e.\

\begin{equation} \label{eq:qindex}
q_i = \frac{1}{N} \sum_{j \neq i} \theta(k_i - k_j),
\end{equation}
where $\theta(x)=0$ for $x\le 0$ and $\theta(x)=1$ for $x>0$. The variable defined in \eqref{eq:qindex} clearly represents a suitable definition of the social ranking, hence the social status, of a given individual in the network. Thus, a reasonable measure of the individuals' mobility in the social climbing game is given by the change in the quantity defined above over a certain time lag $\Delta t$, i.e.\ $\Delta q_i (\Delta t) = q_i(t+\Delta t) - q_i(t)$, for $i=1,\ldots,N$. 
\begin{figure}
	\centering
  	\includegraphics[width=0.48\columnwidth]{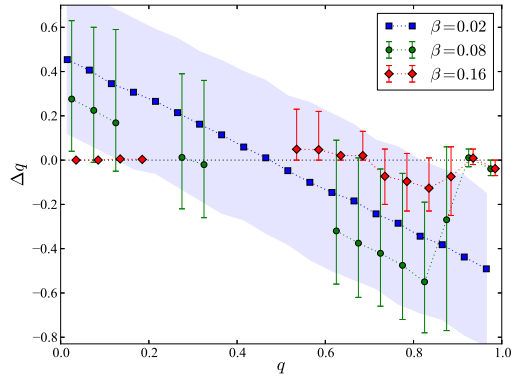}
\caption{(Color online) Relation between the $q$ index defined in \eqref{eq:qindex} and its variation $\Delta q$ over a given time lag $\Delta t = 10^4$ Monte Carlo steps for a network with $N=100$  and $M=1000$. The different curves refer to inverse temperatures $\beta = 0.02$, $0.08$, $0.16$, respectively corresponding to values below, slightly above and well above the critical value for the system under study. Points refer to the average variation $\Delta q$ over an equally spaced grid of $q$ values (going from 0 to 1 in steps of 0.05). Results obtained by averaging over 100 networks. Shaded area (for $\beta = 0.02$) and error bars (for $\beta = 0.08$, $0.16$) represent the central $68\%$ of the events. Points and error bars relative to different values of $\beta$ have been shifted to enhance readability.}
\label{fig:mobility}
\end{figure}

In Fig.\ \ref{fig:mobility} some typical behaviors of the $q$ index defined in \eqref{eq:qindex} are shown. All examples refer to networks with $N=100$ nodes and $M=1000$ links. In such plots, the average of the change $\Delta q$ is shown as a function of $q$, in order to provide information about the typical social mobility over a time lag $\Delta t$ for an agent whose social ranking at the beginning of such a time lag is quantified by $q$. As can be seen, depending on the preference for social status, i.e.\ on the inverse temperature parameter $\beta$, very different situations can happen. In a rather disordered society (low values of $\beta$) the relation between $q$ and $\Delta q$ is clearly linear, and does not depend strongly on the time lag size $\Delta t$. In particular, it can be seen that, on average, individuals sitting at the bottom of the ranking typically end up higher in the social ladder after some time, whereas individuals sitting atop the hierarchy are prevented from keeping their social status intact for a long time. When the preference for social status crosses its critical value, such a picture starts changing quite dramatically. For values of $\beta$ slightly larger than the critical value $\beta_c$ agents with low degrees still have a chance to climb up the social ladder, especially over rather long time lags, whereas the dominant individuals ($q \gtrsim 0.9$) typically get to keep their social ranking. It is worth to remark that for low $\beta$ the distribution of $k_i$ is not very skewed, so changes in social ranking $\Delta q_i$ are more frequent. In this sense, our notion of social mobility captures aspects related to social dynamics but it also depends on the stationary distribution of $q_i$'s, i.e. on the degree of inequality. As $\beta$ increases, the degree distribution acquires skewness, with few individuals having many links and reduced social mobility. For $\beta>\beta_c$ the population separates into two groups, those with $k_i$ of order $N$ and those with very few links, with suppressed mobility across the whole social hierarchy.
Also, as can be seen from the right plot in Fig.\ \ref{fig:mobility}, when the critical threshold is crossed the social network becomes ``fragmented'', as the $q$ index is no longer defined over the whole interval $[0,1]$. In a strongly ordered society, i.e.\ $\beta$ well above its critical value, agents with low degrees are almost completely stuck, and all of the social mobility happens in the top half of the social network, i.e.\ amongst agents with $q > 0.5$, and this is precisely due to the freezing of the dominant individuals inducing social mobility to disappear completely also amongst nodes with small degrees. These results complement, at a ``microscopic'' level, those presented in Fig.\ \ref{fig:Kendall}.

\section{Conclusions} \label{sec:conclusions}

In summary, we have discussed a very simple model for the dynamics of a social network where the agents' quest for high status in the social hierarchy reinforces the latter while reducing social mobility. The model is very stylized and far from a realistic description of social dynamics. Yet, it captures some key ingredients that are enough to reproduce stylized facts known at least since the work of Vilfredo Pareto \cite{Pareto}. Namely, Pareto observed that societies tend to organize in a hierarchical manner, with the emergence of ``social elites''\footnote{A similar concept of ``power elites'' has been discussed in \cite{Mills}.}. Our model, as well as Refs. \cite{koenig09,koenig11}, provides a formal framework showing that individual incentives for high social status are enough to confer this property to the social network, even in the absence of explicit discrimination of particular groups (e.g. cast system or racial segregation) or preferential biases (e.g. hereditary rules). In addition, we find that the hierarchical state is remarkably stable, with suppressed social mobility in the upper and lower parts of the hierarchy. Notably, Pareto himself observed that social mobility is higher in the middle classes \cite{Pareto}. Furthermore, our model exhibits a negative dependence between mobility and inequality, in the sense that more hierarchically structured (i.e. unequal) societies manifest a lower degree of mobility. It is tempting to relate this to the pervasive empirical observation that more unequal societies tend to have lower inter-generational mobility \cite{WilkinsonPickett2010,BjorklundJantti1997,AndrewsLeigh2009}. Our model neglects important dimensions, such as wealth or political power that, however, likely contribute to reinforce our results.

Secondly, we show that the social climbing game admits a potential function, thereby allowing us to deploy techniques and concepts of statistical mechanics to understand the behavior of the system. Statistical mechanics provides a natural language for discussing collective properties of societies. 
For example, the emergence of a social hierarchy in a system of \emph{a priori} identical individuals is an example of spontaneous symmetry breaking, whereby the associated loss of ergodicity accounts for the reduced social mobility. 

The present paper was mostly focused on investigating the model via numerical simulations. The mean field approach discussed in Ref. \cite{Park2004a,Park2004b} is not applicable in our case, because the density of links that plays the role of an order parameter in Ref. \cite{Park2004a,Park2004b} is fixed in our case. Indeed, the phenomenology we find is different from that of Ref. \cite{Park2004a,Park2004b} as we do not find evidence of hysteresis phenomena: there is no range of parameters where the disordered and the ordered societies are both stable. We speculate that this might be related to the fact that in the social climbing game there are mechanisms by which a social hierarchy can ``nucleate'' gradually in an ordered society, by forming social elites that grow over time. 

This and other issues can in principle be addressed within more sophisticated statistical mechanics approaches. In this respect, it is worth to mention that it is possible to map the problem into that of an interacting lattice gas that possibly admits for a full and exact statistical mechanics treatment. Work in this direction is currently in progress.

\begin{acknowledgements}
We thank Sanjeev Goyal and Giacomo Gori for precious hints and fruitful discussions. C.~J.~T. acknowledges financial support from Swiss National Science Foundation through grant  100014\_126865 and SBF (Swiss Confederation) through research project C09.0055. M.~B. heartily thanks M.~V.~Carlucci for her dedicated support.
\end{acknowledgements}

\appendix
\section{Ergodicity of the dynamics} \label{sec:appendix}

Let $\Gamma^{C}(N,M)$ be the space of connected graphs with $N$ vertices and $M$ edges.
In order to prove ergodicity, we have to show that, with a finite number of basic moves, we can reach any connected graph in $\Gamma^{C}(N,M)$, starting from another arbitrary graph in the same set. Before delving into the technical details, we give a simple intuitive sketch of this proof.

For a finite value of $\beta$, the dynamics consists of \emph{reversible} moves as the one depicted in fig.~\ref{fig:move}; such moves  can be thought of as a ``sliding'' of the edge  $e_{ik} $  on the path of length one $(k,j)$ from vertex $k$ to vertex $j$. The key observation to prove the ergodicity by induction is that, since the graph is finite and connected, there always exists a path of minimum length that connects two arbitrarily chosen vertices in the graph. Then, we can proceed in three steps.

\begin{enumerate}
 \item  Let there be two graphs in $\Gamma^{C}(N,M)$  which differ from each other only by an edge incident on the same vertex, $v_k$.
We first prove that by means of basic moves, we can transform one into the other. To do so, it suffices to slide the edge along the path that connects the other end of the edge, which we know to exist because the graph is connected (Prop.~1, Prop.~2).
 \item Let there be two graphs in $\Gamma^{C}(N,M)$ that differ by an edge with arbitrary ends. By applying the previous step twice, we show that there exists a finite set of moves that allows us to reach
 one configuration starting from the other (Prop.~3).
 \item Finally, let there be two arbitrary graphs in $\Gamma^{C}(N,M)$.
 Moving one edge at time, we show by induction that it is possible to reach one graph starting from the other with a finite number of moves. Thus, the ergodicity is proved (Prop.~4).
\end{enumerate}

We now proceed with the detailed Proof.

\begin{figure}
	\centering
  	\includegraphics[width=0.65\columnwidth]{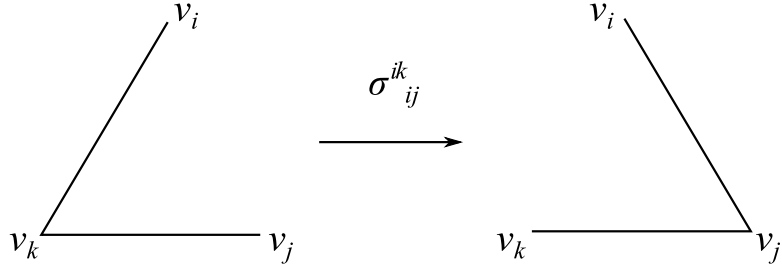}
\caption{The rewiring move ($c$-swap) of the dynamics.}
\label{fig:move}
\end{figure}

\begin{definition}[$c$-swap]
Let us choose a labeling for the space of vertices $V=\{ v_1,\ldots,v_N \}$ and an induced labeling for the edges $E=\{e_{ij}\}$ where $e_{ij}=e_{ji}=(v_i,v_j)$ denotes the undirected edge between $v_i$ and
$v_j$. Let us define a  transformation $\sigma_{ij}^{ik}: \Gamma^C(N,M) \mapsto \Gamma^C(N,M) $, called \emph{corner swaps ($c$-swaps)}, as following
\begin{equation}
 \sigma_{ij}^{ik}(\Gr)=\Gr'=(V,E')
\end{equation}
such that
 \begin{equation}
 E'=\begin{cases}
	\left ( E \setminus \{ e_{ik} \} \right ) \cup \{ e_{ij} \}  & \text{ if } (e_{kj}, e_{ik}\in E)\wedge( e_{ij} \notin E)  \\
     E & \text{ otherwise }
     \end{cases}.
 \end{equation}
\end{definition}

\begin{proposition}\label{pr:oneedge_0}
Let $\Gr=(V,E)$ and $\Gr'=(V,E')$ be two graphs in $\Gamma^C(N,M)$  that differ by an edge incident on the same vertex, i.e.  $|E|=|E'|=M$,  $|E'\cap E|=M-1$, $E\setminus E'=\{ e_{ik} \}$ and $E'\setminus E=\{ e_{ij} \}$, and such that the shortest  path $P$ from $v_k$ to $v_j$ does not contain neither $v_i$ nor any of its neighbors.

There exists an integer $l$ and a finite sequence of graphs in $\Gamma^C(N,M)$, $\Gr^n$  such that:
\begin{enumerate}[(i)]
\item $\Gr=\Gr^0$ and $\Gr'=\Gr^l$.
\item For all $0\leq n < l$ there exist adjacent vertices $v_{k_n}$,$v_{k_n+1}$ such that $\Gr^{n+1}=\sigma_{i k_{n+1}}^{i k_n } (\Gr^n)$, where $k_0=k$ and $k_l=j$.
\end{enumerate}
\end{proposition}
\begin{proof}
Let $l$ be the length of $P$.

Let $v_{k_1}$ be the unique neighbor of $v_k$ that lies in $P$. If we set $\Gr^{1}=\sigma_{i k_1}^{i k} ( \Gr )$, the $c$-swap reduces the distance between $v_i$ and $v_j\,$, since the neighbor of $v_k$ that lies in $P$ must have a distance $l-1$ from $v_j$.
We reiterate the procedure on $\Gr^1$ and obtain in such a way a sequence of graphs that satisfies property (ii).
 Now, since at any step the length of $P$ diminishes by 1, after the $l$-th step, in the graph  $\Gr^{l}$  $v_i$ and $v_j$ will be neighbors.
 Thus, since no other edge was changed by applying $c$-swaps,  $\Gr^l=\Gr'$ proving property (i).
\end{proof}

\begin{proposition}\label{pr:oneedge_s}
Let $\Gr=(V,E)$ and $\Gr'=(V,E')$ be two graphs in $\Gamma^C(N,M)$  which differ by an edge incident on the same vertex, i.e.  $|E|=|E'|=M$, $|E'\cap E|=M-1$, $E\setminus E'=\{ e_{ik} \}$ and $E'\setminus E=\{ e_{ij} \}$.

There exists an integer $l$ and a finite sequence of graphs in $\Gamma^C(N,M)$, $\Gr^n$  such that:
\begin{enumerate}[(i)]
\item $\Gr=\Gr^0$ and $\Gr'=\Gr^l$
\item For all $0\leq n \leq l$ there exist adjacent vertices $v_{k_n}$,$v_{k_n+1}$ such that $\Gr^{n+1}=\sigma_{i k_{n+1}}^{i k_n} (\Gr^n)$.
\end{enumerate}
\end{proposition}
\begin{proof}
Let $P$ be the shortest path in $\Gr$ from $v_k$ to $v_j$ that does not contain $(v_k,v_i)$.

There are four possible cases :
\begin{enumerate} [(i)]
\item $P$ does not contain neither $v_i$ nor any of its neighbors other than $v_k$. The thesis is proven applying proposition~\ref{pr:oneedge_0} directly to $P$.

\item $P$ contains $v_i$.  Let $P_1$ be the shortest path from $v_k$ to $v_i$ that does not contain  the edge $(v_k,v_i)$. let $P_2$ be the shortest path from $v_i$ to $v_j$, clearly $P=P_1 \oplus P_2$, where $\oplus$ is the path concatenation. Since by construction there  are no neighbors of $v_k$ in $P_2$ (otherwise $P$ would not contain $v_i$)  we can apply proposition~\ref{pr:oneedge_0}  and reach $\Gr''=(V,(E\setminus \{e_{ki}\})\cup\{e_{kj}\})$; on the other hand there cannot be neighbors of $v_j$ in $P_1$ (otherwise there would be a shortest path not containing $v_i$) and thus applying again proposition~\ref{pr:oneedge_0} along $P_1$ we reach $\Gr'$ proving the thesis.
\item  $P$ does not contain $v_i$ but two of its neighbors, $c$ and $f$ such that $c \neq v_k$, $f\neq v_k$ and $|c,v_k| < |f,v_k|$, where $| \cdot , \cdot |$ represents the graph distance between two vertices.
We first note that $c$ and $f$ must be neighbors, otherwise $P$ should include $v_i$.
Then, as in case (ii), by minimality we can write $P=P_1\oplus (c,f)\oplus  P_2$ where $P_1$ is the shortest path from $v_k$ to $c$ and $P_2$ is the shortest path from $f$ to $v_j$. It is easy to see  that $Q_2=(v_i,f)\oplus P_2$ is a shortest path from $v_i$ to $v_j$: if it were not so, there would exist a path $Q_2'$ from $v_i$ to $v_j$ strictly shorter than $Q_2$, but in that case $P_1\oplus (c,v_i)\oplus Q_2'$ would be a shortest path from $v_k$ to $v_j$ containing $v_i$, in contradiction with our hypotheses. A similar argument holds for $Q_1$. As before, since, by minimality, there cannot be neighbors of $v_k$ in $P_2$, it is possible to reach the graph $\Gr''=(V,(E\setminus \{e_{ki}\})\cup\{e_{kj}\})$ by applying proposition~\ref{pr:oneedge_0} to $Q_2$; since by minimality there cannot be neighbors of $v_j$ in $P_1$, we can apply proposition~\ref{pr:oneedge_0} to $\Gr''$ along $Q_2$ and reach $\Gr'$ proving the thesis.

\item The shortest path $P$ contains only one neighbor of $v_i$ other than $v_k$, let us call it $m$. As before, $P=P_1\oplus P_2$ where $P_1$ is the shortest path from $v_k$ to $m$ and $P_2$ is the shortest path from $m$ to $v_j$. Since by construction there cannot be other neighbors of $i$ in $P_2$,  we can apply proposition~\ref{pr:oneedge_0} to $P_2$ and reach the graph $\Gr^*=(V,(E\setminus \{e_{im}\})\cup\{e_{ij}\}$. On the other hand, by construction there cannot be neighbors of $v_i$ in $P_1$ other than $v_k$ and thus we can apply proposition~\ref{pr:oneedge_0} to $P_2$ and reach $\Gr'$ proving the thesis.
\end{enumerate}

\end{proof}

\begin{proposition}\label{pr:oneedge}
Let $\Gr=(V,E)$ and $\Gr'=(V,E')$ be two graphs in $\Gamma^C(N,M)$ such that 
$|E|=|E'|=M$ and $|E \cap E'|=M-1$. Let us assume that, in particular, $E=\{ e_{ij} \} \cup (E \cap E')$ and $E'=\{ e_{hk} \} \cup (E \cap E')$.

Thus there exists an integer $l$ and a finite sequence of graphs in $\Gamma^C(N,M)$, $\Gr^n$  such that:
\begin{enumerate}[(i)]
\item $\Gr=\Gr^0$ and $\Gr'=\Gr^l$
\item For all $0\leq n < l$ there exist adjacent vertices $v_{k_n}$,$v_{k_n+1}$ such that $\Gr^{n+1}=\sigma_{i k_{n+1}}^{i k_n} (\Gr^n)$.
\end{enumerate}
\end{proposition}
\begin{proof}
Let us define the graph $\Gr''=(V,E'')$ such that $E''=(E \setminus \{ e_{ij} \})\cup \{e_{ih}\}$. Applying proposition~\ref{pr:oneedge_s} first to graphs $\Gr$ and $\Gr''$ and then to graph $\Gr''$ and $\Gr'$ proves the thesis.
\end{proof}

\begin{definition}[$g$-swap]
Let $\Gr=(V,E)$ and $\Gr'=(V,E')$ be two graphs in $\Gamma^C(N,M)$ which differ at most by an edge, that is such that
$|E|=|E'|=M$ and $|E \cap E'|=M-1$. Let us assume that, in particular, $E=\{ e_{ij} \} \cup (E \cap E')$ and $E'=\{ e_{hk} \} \cup (E \cap E')$.

We define a \emph{global swap or $g$-swap} of the edge $e_{ij}$ to the edge $e_{hk}$ a transformation such that:
\begin{equation}
\Gr'=\Sigma_{ij}^{hk}(\Gr)
\end{equation}
\end{definition}
Proposition~\ref{pr:oneedge} simply states that any global swap can be obtained as the composition of a minimal set of corner swaps between adjacent vertices.

\begin{proposition}\label{pr:ergodicity}
Let $\Gr=(V,E)$ and $\Gr=(V,E')$ be two graphs in $\Gamma^C(N,M)$.
There exists an integer $d$ and a sequence of graphs $\Gr^n(V, E_n)$ in $\Gamma^C(N,M)$ such that:
\begin{enumerate}[(i)]
\item $\Gr=\Gr^0$ and $\Gr'=\Gr^d$
\item For all $0\leq n < d $ there exist four vertices $v_i\,$, $v_j\,$,$v_h$ and $v_k$ such that $\Gr^{n+1}=\Sigma_{ij}^{hk}(\Gr^n)$
\end{enumerate}
 \end{proposition}
\begin{proof} Let $\mathcal{Z}=(V,Z=E\cap E')$, and let us define $\delta=|Z|$. We proceed by induction on the number $\delta$.
\begin{description}
\item[{\bf Base case}] If $\delta=M-1$, the Thesis is trivially true because of Proposition \ref{pr:oneedge}.
\item[{\bf Inductive step}] Let us assume that the Thesis holds for $\delta = M-d$, we want to show that this implies that it also holds for $\delta = M - d - 1$, with $d < M -1$. Let us assume that $\Gr = (V,E)$ and $\Gr' = (V,E')$ are such that $ | E' \cap  E | = M - d -1$. Let $e_{ij} \in E \setminus ( E \cap E' ) $ and $e_{hk} \in E' \setminus ( E \cap E' )$. Moreover, let $E'' = (E \setminus \{e_{ij}\} ) \cup \{ e_{hk} \}$.
By construction, $| E \cap E''| = M -1$ and $ | E' \cap E''| = M-d$. Finally, let $\Gr''= (V,E'')$.  Since $\Gr''$ and $\Gr'$ differ by $M-d$ edges, by inductive assumption there exists a sequence $\Gr^{i}$, with $i \in [0,d]$, such that $\Gr^0 = \Gr'$ and $\Gr^d = \Gr''$, that satisfies the Thesis. Moreover, by Proposition \ref{pr:oneedge}, there exists a $g$-swap such that $\Gr = \Sigma_{hk}^{ij}(\Gr'')$. Thus, the complete sequence $\Gr'=\Gr^0,\Gr^1,\cdots,\Gr''=\Gr^{d},\Gr=\Gr^{d+1}$ satisfies the Thesis.
\end{description}
\end{proof}
Proposition~\ref{pr:ergodicity} and proposition~\ref{pr:oneedge} state simply that any two connected graphs with the same number of edges can be obtained one from the other
applying a finite sequence of $c$-swaps. Moreover, since the number of edges is finite, then there must be a minimal sequence of $c$-swaps that connects any two of such graphs.
Since, for finite $\beta$, all $c$-swaps are allowed with nonzero probability, this proves the ergodicity. $\square$

\bibliographystyle{spmpsci}      
\bibliography{collection}

\end{document}